\documentclass[12pt]{article}

\usepackage{amsmath,amsthm,amssymb}
\usepackage{graphicx,psfrag,epsf}
\usepackage{enumerate}
\usepackage{natbib}

\newcommand{\blind}{0}

\usepackage[T1]{fontenc}
\usepackage{url}
\usepackage{booktabs}
\usepackage{amsfonts}
\usepackage{nicefrac}
\usepackage{microtype}
\usepackage{hyperref}
\usepackage{array}
\usepackage{graphics}
\usepackage{listings}
\usepackage{tabu}
\usepackage{bm}
\usepackage[table]{xcolor}
\usepackage{color}
\usepackage{colortbl}
\usepackage{mathrsfs}
\usepackage{accents}
\usepackage{tikz}
\usetikzlibrary{shapes}
\usepackage{xspace}
\usepackage{bbm}
\usepackage{caption}
\usepackage{subcaption}
\usepackage{placeins}
\usepackage{multirow}
\usepackage{arydshln}
\usepackage[utf8]{inputenc}

\newtheorem{proposition}{Proposition}

\begin{document}

\def\spacingset#1{\renewcommand{\baselinestretch}%
{#1}\small\normalsize} \spacingset{1}


\if0\blind
{
  \title{\bf \emph{CavMerge}: Merging K-means Based on Local Log-Concavity}
  \author{
    Zhili Qiao\thanks{
      Zhili Qiao is Ph.D.\ Candidate, Department of Statistics, Iowa State University, Ames, IA 50010 (E-mail: zlqiao819@gmail.com).
      Wangqian Ju is Ph.D.\ Candidate, Department of Statistics, Iowa State University, Ames, IA 50010 (E-mail: wju@iastate.edu).
      Peng Liu is Associate Professor, Department of Statistics, Iowa State University, Ames, IA 50010 (E-mail: pliu@iastate.edu).
    }\hspace{.2cm}\\
    Department of Statistics, Iowa State University\\
    and \\
    Wangqian Ju \\
    Department of Statistics, Iowa State University\\
    and \\
    Peng Liu \\
    Department of Statistics, Iowa State University}
  \maketitle
} \fi

\if1\blind
{
  \bigskip
  \bigskip
  \bigskip
  \begin{center}
    {\LARGE\bf \emph{CavMerge}: Merging K-means Based on Local Log-Concavity}
  \end{center}
  \medskip
} \fi

\bigskip

\begin{abstract}
K-means clustering, a classic and widely-used clustering technique, is known to exhibit suboptimal performance when applied to non-linearly separable data. Numerous adjustments and modifications have been proposed to address this issue, including methods that merge K-means results from a relatively large K to obtain a final cluster assignment. However, existing methods of this nature often encounter computational inefficiencies and suffer from hyperparameter tuning. Here we present \emph{CavMerge}, a novel K-means merging algorithm that is intuitive, free of parameter tuning, and computationally efficient. Operating under minimal local distributional assumptions, our algorithm demonstrates strong consistency and rapid convergence guarantees. Empirical studies on various simulated and real datasets demonstrate that our method yields more reliable clusters in comparison to current state-of-the-art algorithms.
\end{abstract}

\noindent%

\vfill

\spacingset{1.45}

\section{Introduction}
\label{sec:intro}

Cluster analysis stands as a cornerstone method in data mining and unsupervised machine learning, providing valuable insights by identifying patterns, structures, or relationships within intricate datasets. In recent years, the application of cluster analysis has proliferated across various disciplines, including but not limited to bioinformatics, finance, marketing, and social network analysis.

The K-means clustering algorithm \citep{MacQueen1967, Lloyd1982, Hartigan1979} has been recognized as one of the most efficient and widely-used clustering techniques. Given a dataset \(X_{n \times p}\) and a predetermined number of clusters \(K\), it assigns data points to clusters based on the proximity to the closest center, using Euclidean distance as the measure. Over time, this method has been refined and extended in various ways, including Bayesian reformulations \citep{Kulis2012, Kurihara2009}, improvements in initializations \citep{Arthur2007, Celebi2013}, and advancements in feature selection \citep{Witten2010, Zhang2020}.

Consider a dataset \(X_{n \times p}\) with \(n\) observations and \(p\) features. A K-means clustering with \(K\) clusters iterates the following two steps until convergence.

\paragraph{K-means clustering.}
\begin{enumerate}
    \item Assign each observation \(x_i, i = 1,...,n\) to the cluster center (mean of all points in the cluster) with minimum Euclidean distance: \(argmin_{1 \leq k \leq K} ||X_i - \mu_k||\);

    \item Recalculate cluster centers \(\mu_k = \frac{\sum_{i \in C_k} X_i}{|C_k|}, k =1,...,K\), where \(C_k\) is the set of all observations in the \(k^{th}\) cluster, and \(|C_k|\) is its cardinality.
\end{enumerate}

While the K-means clustering algorithm is computationally efficient and straightforward to interpret, it does harbor two significant limitations:

\begin{itemize}
    \item The output clusters from the K-means algorithm are convex sets.

    The K-means algorithm segments the \(R^p\) space into \(K\) exhaustive and exclusive convex subspaces: \(\forall x, y \in C_k, \forall \alpha \in [0, 1], \alpha x + (1-\alpha)y \in C_k\).

    This implies that K-means is primarily suited to handle linearly separable data.
    
    \item K-means exhibits sub-optimal performance with unbalanced data.
    
    Given that the K-means algorithm relies on minimizing Euclidean distance to cluster centers, it has been observed that K-means can perform poorly when the dataset contains unbalanced true classes \citep{Jain1999, Chiang2010}.

\end{itemize}

Several approaches have been proposed to mitigate these issues. For example, the application of kernel methods has been suggested to handle non-convexity \citep{Dhillon2004, Tzortzis2009}, while the use of weighted K-means has been advocated for adjusting unequal-sized clusters \citep{Zhang1996, Lu2021}. Nevertheless, most of these approaches are case-specific in determining kernels or weights and often require additional efforts in hyperparameter tuning.

Recently, a new approach called syncytial clustering has been proposed. It utilizes K-means clustering results as initial inputs and then adopts a merging process to obtain the final cluster assignments.

\paragraph{Merging K-means.} Merging results from K-means clustering allows for the formation of clusters of varying shapes, thereby addressing the two primary limitations of K-means. This type of method starts from a K-means result with a relatively large \(K\), and uses some methods to determine the merging status of pairs of clusters. Each cluster in the final result is a combination of one or more initial K-means clusters. \citet{Peterson2018} proposed to use hierarchical clustering on K-means results; \citet{AlmodovarRivera2020} proposed utilizing an overlap-based estimate to determine the merging status of K-means clusters; \citet{Wei2023} proposed merging based on a dimension-free surrogate density measure.

Beyond K-means, merging algorithms have also become prominent within model-based clustering frameworks. For instance, \citet{Baudry2010} proposed a methodology for merging components within Gaussian mixtures, while \citet{Melnykov2016} advocated for the fusion of model components based on pairwise overlap.

\paragraph{Our contribution.}

Here we propose a novel method for merging K-means results based on local log-concavity and call it the \emph{CavMerge} algorithm. Our method is non-iterative and involves no estimation of density or distribution of any kind. Like many other K-means merging algorithms, CavMerge estimates a merging score for each pair of clusters. In contrast to the intricate kernel smoothing processes utilized in \citet{Peterson2018} and \citet{AlmodovarRivera2020}, our method incorporates a technique that is significantly more straightforward and computationally efficient. In addition, our method maintains a strong theoretical foundation and achieves superior clustering performance.

After we developed our algorithm, we noticed that a recently proposed \emph{Skeleton Clustering} algorithm by \citet{Wei2023} shares some similarity with our method. However, our motivation and fundamental idea come from a completely different perspective. Instead of a density measure, our method utilizes a fully non-parametric pairwise similarity measure that is scale-free, more intuitive, and easier to implement. We will discuss this difference in more detail in Section \ref{sec:conc} and the Appendix. 

The rest of this paper is structured as follows. In Section \ref{sec:notation}, we introduce formal notations and definitions of K-means. In Section \ref{sec:meth}, we describe the intuition of this method and give some theoretical background. In Section \ref{sec:algo}, we specify and explain our merging algorithm. In Section \ref{sec:experiments}, we provide some experimental results to show the effectiveness of our method. Finally in Section \ref{sec:conc}, we conclude and discuss some possible extensions.

\section{Notations}
\label{sec:notation}

In this section, we describe some notations about K-means clustering for use throughout our paper. 

Let \(X_{n \times p} = (x_1,...,x_n)'\) denote the data matrix with \(n\) observations and \(p\) features, each \(x_i \in R^p\) is an observation in \(p\) dimensional feature space. We assume that those \(n\) data points come from a mixture of \(G\) true classes, each with probability density function \(f_g(\cdot)\): \(x_i \sim \sum^G_{g=1} \pi_g f_g (y), \pi_g > 0, \sum^{G}_{g=1} \pi_g = 1\) . We are interested in clustering \(x_i\)'s into several exclusive and exhaustive clusters. 

For a standard K-Means clustering result on \(X_{n \times p}\) with \(K\) clusters: 

\begin{itemize}
    \item Let \(C_k = \{x_i:x_i\in\ cluster\ k\}\) denote all the observations \(x_i\)'s in cluster \(k\), \(k = 1,...,K\);

    \item Let \(\mu_k \in R^p\) denote the center of cluster \(k\): \(\mu_k = \frac{1}{|C_k|} \sum_{i \in C_k} x_i\) 

    \item For each pair of clusters \((C_{k_1}, C_{k_2}), k_1 \neq k_2\), we define its \textbf{decision boundary} \(D_{k_1 k_2}\) as the \((p-1)\) dimensional hyperplane that is equal-distant to their cluster centers: \(D_{k_1 k_2} = \{ x \in R^p: ||x - \mu_{k_1}|| = ||x - \mu_{k_2}|| \}\).
    
    \item A decision boundary \(D_{k_1 k_2}\) is called a \textbf{nontrivial decision boundary}, if there exists \(x \in D_{k_1 k_2}\) such that \(||x - \mu_{k_1}|| = ||x - \mu_{k_2}|| < ||x - \mu_{k'}||, \forall k' \neq k_1, k_2\).

    \item We call a pair of clusters \((C_{k_1}, C_{k_2})\) an \textbf{adjacent pair}, if their share a nontrivial decision boundary.
\end{itemize}

\section{Methodology and Theoretical Background}
\label{sec:meth}

In this section, we introduce our main idea and provide some theoretical background for our algorithm that will be presented in the next section. The main idea comes from \textbf{log-concavity}, a property that many commonly-used distributions satisfy. We say a distribution has a log-concave probability density/mass function if its logarithm is concave, i.e.,

\begin{align*}
    \gamma \log f(x) + (1-\gamma) \log f(y) \leq \log f(\gamma x + (1-\gamma) y), \forall x,y \in R^p, \forall \gamma \in [0, 1]
\end{align*}

Log-concavity is a very weak distributional assumption. Many commonly-used distributions are log-concave, see \citet{Bagnoli2005}.

In the above inequality, setting \(\gamma = \frac{1}{2}\) and taking exponential on both sides yields \(f(x)f(y) \leq f^2(\frac{x+y}{2})\). If we replace \((x,y)\) by \((x-a, x+a)\) and perform a double integral on \(x \in Q\) for some \(Q \subset R^p\), this further yields (if denominator is non-zero) 

\begin{align}
    \frac{[\int_Q f(x) dx]^2}{\int_{Q} f(x-a) dx \int_{Q}f(x+a) dx} = \frac{[\int_Q f(x) dx]^2}{\int_{Q-a} f(x) dx \int_{Q+a} f(x) dx} \geq 1, \forall x,a \in R^p, \forall Q \subset R^p
\end{align}

Now following the cluster analysis framework, suppose we have a data matrix \(X_{n \times p}\) and its clustering result \((C_1,..., C_K)\) from a standard K-means clustering algorithm. For two adjacent clusters \(C_{k_1}, C_{k_2}\) with centers \(\mu_{k_1}, \mu_{k_2}\), consider \(Q-a\) as a subspace surrounding \(\mu_{k_1}\), where \(2a = \mu_{k_2} - \mu_{k_1}\). Then it is clear that the two terms in the denominator of~(1) are the cumulative densities around cluster centers \(\mu_{k_1}\) and \(\mu_{k_2}\), and the numerator is the square of cumulative density around their decision boundary. A discrete version of this fraction, the ratio of the number of observations, may serve as a score to identify the merging status of two clusters. 

\paragraph{Theorem 1.} Assume that a \(p\)-dimensional random variable \(X\) has log-concave pdf \(f(\cdot)\), and let \(x_1, ..., x_n \in R^p\) denote a realization of this distribution with \(n\) observations. For any subspace \(Q \subset R^p\) and any vector \(a \in R^p\), let

\begin{align*}
    m_1 = \sum^n_{i=1} 1(x_i+a \in Q),\quad m_2 = \sum^n_{i=1} 1(x_i \in Q),\quad m_3 = \sum^n_{i=1} 1(x_i-a \in Q)
\end{align*}

denote the number of observations in the three subspaces \(Q-a, Q, Q+a\), respectively. Furthermore, assume that the cumulative densities on those three subspaces are all positive. Then we have the following: 

\begin{enumerate}
    \item[(a)] For any given \(\epsilon,\delta > 0\) and if \(m_1 m_3 > 0\), there exists a \(N\) such that for all \(n>N\), we have \(P(\frac{m^2_2}{m_1 m_3} > 1-\epsilon) > 1-\delta\).
    
    \item[(b)] The convergence rate of \(\frac{m^2_2}{m_1 m_3}\) to \(C = \frac{(\int_{U} f(x) dx)^2}{\int_{U-a} f(x) dx \int_{U+a} f(x) dx}\) is \(O(\frac{1}{\sqrt{n}})\).
\end{enumerate}

The proof of Theorem 1 is provided in Appendix~\ref{app:proof_thm1}.

Theorem 1 above provides the convergence guarantee of ``ratio of points'' to ``ratio of densities'', which is the foundation of our \emph{CavMerge} algorithm.

\section{\emph{CavMerge} Algorithm Based on Log-Concavity}
\label{sec:algo}

In this section, we propose \emph{CavMerge}, an algorithm for merging K-means results based on local log-concavity. This is a non-iterative algorithm that does ``nothing but counting points''. This algorithm aims to detect pairs of clusters that are adjacent and are from the same distribution. It determines pairwise merging status by computing a score based on the local concentration of observations around cluster centers and around boundaries. In the end, a merging status matrix is formed, and we merge all positive pairs to get the final cluster assignment.

First, we specify two notations for computation:

\begin{itemize}
    \item For two \(p\) dimensional vectors \(x\) and \(y\), let \(P_y(x) = \frac{|x^T y|}{||y||}\) denote the projection length of \(x\) on \(y\);

    \item For a \(p\) dimensional point \(x\) and a \(1\) dimensional line \(l_{ab}\) determined by two different points \(a,b \in R^p\), let \(d(x, l_{ab}) = \sqrt{||x - a||^2 - \frac{|(x-a)^T(b-a)|^2}{||b-a||^2}}\) denote the projection distance from \(x\) to \(l_{ab}\).
\end{itemize}

Next, we propose our algorithm that consists of 5 steps and then explain each step in the following subsections.

\paragraph{\emph{CavMerge} algorithm: merging K-means based on local log-concavity.}

\begin{enumerate}
    \item[Step 1:] \emph{Initialize.}
    
    Set \(K_{max} = \max(\lfloor n \rfloor, 30)\). For each \(K = 1,2,..., K_{max}\), perform standard K-means clustering on data matrix \(X_{n \times p}\). Use jump statistic \citep{Sugar2003} to identify the optimal \(K\) and its corresponding cluster assignments \(\{C_{1}, ..., C_{K}\}\)

    \item[Step 2:] \emph{Identify Adjacent Clusters.}

    For each observation \(X_i, i = 1,...,n\), calculate its distance to all initial cluster centers \(\mu_1,...,\mu_K\). Find the first two minimum distances, denote the corresponding clusters as \((C_{k_1(i)}, C_{k_2(i)}), i = 1,...,n\)
    
    We identify \(A = \{\cup^n_{i=1} (C_{k_1(i)}, C_{k_2(i)})\}\) as the collection of all adjacent cluster pairs.

    \item[Step 3:] \emph{Calculate Merging Scores.} 

    For each adjacent pair of clusters \((C_{k_1}, C_{k_2}) \in A\):

    \begin{enumerate}
        \item Calculate their centers \(\mu_{k_1}, \mu_{k_2}\); let \(l_{\mu_{k_1}\mu_{k_2}}\) denote the 1-dimensional line determined by \(\mu_{k_1}\) and \(\mu_{k2}\).
    
        \item For each point \(x_i \in C_{k_1} \cup C_{k_2}\), calculate its distance \(d(x_i, l_{\mu_{k_1}\mu_{k_2}})\) to \(l_{\mu_{k_1}\mu_{k_2}}\). Let \(r_{k_1 k_2} = \max_{x_i \in C_{k_1} \cup C_{k_2}} d(x_i, l_{\mu_{k_1} \mu_{k_2}})\).
    
        \item Find the number of points \(m_1, m_2, m_3\) that satisfy the following conditions:
        
        \begin{align*}
            m_1 &= \sum^n_{i=1} 1 \{d(x_i, l_{\mu_{k_1}\mu_{k_2}}) < r_{k_1 k_2}, P_{\mu_{k_2} - \mu_{k_1}}(x_i - \mu_{k_1}) < ||\frac{\mu_{k_2} - \mu_{k_1}}{4}||\} \\
            m_2 &= \sum^n_{i=1} 1 \{d(x_i, l_{\mu_{k_1}\mu_{k_2}}) < r_{k_1 k_2}, P_{\mu_{k_2} - \mu_{k_1}}(x_i - \frac{\mu_{k_1} + \mu_{k_2}}{2}) < ||\frac{\mu_{k_2} - \mu_{k_1}}{4}||\} \\
            m_3 &= \sum^n_{i=1} 1 \{d(x_i, l_{\mu_{k_1}\mu_{k_2}}) < r_{k_1 k_2}, P_{\mu_{k_2} - \mu_{k_1}}(x_i - \mu_{k_2}) < ||\frac{\mu_{k_2} - \mu_{k_1}}{4}||\}
        \end{align*}
    
        \item Calculate the score \(s_{k_1 k_2} = \frac{m^2_2}{m_1 m_3}\)
    \end{enumerate}

    Output a \(K \times K\) symmetric merging score matrix \(S\) (set diagonal elements to \(\infty\) and scores for non-adjacent pairs to \(0\)).

    \item[Step 4:] \emph{Adjust for Corner Cases.}

    For each cluster \(C_k\) with \(|C_k| \leq 3\), find its closest neighboring cluster \(C_{k'}: argmin_{k' \neq k} ||\mu_k - \mu_{k'}||\). Set \(s_{kk'} = \infty\).
    
    \item[Step 5:] \emph{Output Final Clusters.} 

    Use \(\frac{1}{s_{k_1 k_2}}\) as the distance metric and single linkage to construct a hierarchical clustering dendrogram; cut the tree at the desired cluster number; perform merging on the initial clusters according to this hierarchical clustering tree.

    Return the final cluster assignments.
\end{enumerate}

\paragraph{Corollary 1.} Assume that in each hyper-cylinder \(H_{k_1 k_2} = \{ y \in R^p: d(y, l_{\mu_{k_1} \mu_{k_2}}) < r_{k_1 k_2}, P_{\mu_{k_2} - \mu_{k_1}}(y - \frac{\mu_{k_1} + \mu_{k_2}}{2}) < ||\frac{3}{4} (\mu_{k_2} - \mu_{k_1})|| \}\), the overall probability density function is log-concave. Then as \(n \to \infty\), we have \(P(S_{k_1 k_2} > 1) \to 1\) for all pairs of \((C_{k_1}, C_{k_2})\) and the score \(S_{k_1 k_2}\) as defined in Step 3 of the algorithm.

\begin{proof}

This serves as a direct result of Theorem 1, by plugging in \(Q  = \{y \in R^p: d(y, l_{\mu_{k_1}\mu_{k_2}}) < r_{k_1 k_2}, P_{\mu_{k_2} - \mu_{k_1}}(y - \frac{\mu_{k_1} + \mu_{k_2}}{2}) < ||\frac{\mu_{k_2} - \mu_{k_1}}{4}||\}\), and \(a = \frac{\mu_{k_2} - \mu_{k_1}}{2}\). 

\end{proof}

As outlined in Corollary 1, the underlying foundation of our algorithm is predicated on the assumption of local log-concavity. Log-concavity is a very weak assumption, there are a lot of commonly used distributions satisfying it \citep{Bagnoli2005}. Furthermore, we only assume local pairwise log-concavity in our algorithm. This means many real-life data can have local approximations satisfying this assumption. Specifically, we want to demonstrate that this log-concave assumption is weaker than assumptions in many other K-means merging algorithms. For example, \cite{AlmodovarRivera2020} utilizes an overlap-based pairwise Gaussian kernel smoothing around the decision boundary, where the log-concavity condition can be seen as automatically holds since the Gaussian pdf is log-concave.

\subsection{Step 1: Initialization}

We start with a large \(K_{max}\), and use the jump statistic \citep{Sugar2003} to find the cluster number with maximum decrease in distortion. This is considered a ``best'' K-means results out of these \(K= 1,..., K_{max}\) in terms of spherical dispersion. We set the power of the distortion equal to \(\frac{p}{2}\) following \citep{Sugar2003}.

\subsection{Step 2: Identification of Adjacent Cluster Pairs}

By the classical duality between Voronoi diagrams and Delaunay triangulations \citep{Voronoi1908}, two cluster centers \(\mu_{k_1}\) and \(\mu_{k_2}\) are \textbf{adjacent} (i.e., their Voronoi cells share a \((p-1)\)-dimensional boundary) if and only if they are connected by an edge in the Delaunay triangulation of \(\{\mu_1, \ldots, \mu_K\}\). Identifying all such edges exactly in high-dimensional spaces is computationally intractable \citep{Polianskii2020}. The trivial upper bound on the number of adjacent pairs is \(\frac{K(K-1)}{2}\), the total number of distinct cluster pairs, which grows \emph{quadratically} with \(K\). We claim that this bound can be tightened:

\begin{proposition}[Linear adjacency bound]
\label{prop:linear_adjacent}
Under the manifold hypothesis that the data concentrate near a \((p-1)\)-dimensional submanifold of \(\mathbb{R}^p\), the number of adjacent cluster pairs is \(O(K)\) instead of \(O(K^2)\).
\end{proposition}

\begin{proof}[Proof sketch]
The \(K\) cluster centers, each being the empirical mean of a subset of data points, lie approximately on the same submanifold as the data. The number of adjacent pairs equals the number of edges in the Delaunay triangulation of the \(K\) centers. For \(p = 2\), by Euler's polyhedral formula the Delaunay triangulation has at most \(3K - 6 = O(K)\) edges. For \(p \geq 3\), \citet{AttaliLinear2004} prove that \(K\) points uniformly sampled on a fixed \((p-1)\)-dimensional polyhedral surface in \(\mathbb{R}^p\) have at most
\[
    \left(1 + \frac{C_S \kappa}{2} + 5300\pi\kappa^2 \frac{L_S^2}{A_S}\right) K
\]
Delaunay edges, where \(C_S\) is the number of facets, \(A_S\) is the total area, \(L_S\) is the total boundary perimeter of the surface, and \(\kappa\) is a uniform density ratio---all constants fixed by the data geometry and independent of \(K\). This is an explicit \(O(K)\) bound. More generally, for data near a \(q\)-dimensional polyhedron in \(\mathbb{R}^p\) (\(2 \leq q \leq p-1\)), \citet{Amenta2007} establish a Delaunay complexity of \(O(K^{(p-1)/q})\), which reduces to the linear \(O(K)\) when \(q = p-1\). A complete proof is provided in Appendix~\ref{app:proof_prop1}.
\end{proof}

This linear adjacency bound is fundamental to the computational efficiency of \emph{CavMerge}: it justifies computing merging scores for only \(O(K)\) cluster pairs rather than all \(O(K^2)\) pairs, yielding substantial savings especially when \(K\) is large.

In practice, since computing the exact Delaunay triangulation is intractable in high dimensions \citep{Polianskii2020}, we adopt the \emph{approximate Delaunay triangulation} introduced in \citet{Wei2023}: for each observation \(x_i\), we find its two nearest cluster centers and mark the corresponding pair as adjacent. This computation costs \(O(nK)\) time, introduces no false positives (non-adjacent pairs cannot be misidentified as adjacent), and any false negatives correspond to cluster pairs with no data point near their shared boundary (i.e., the pairs least likely to warrant merging). A full justification is provided in Appendix~\ref{app:proof_prop1}.

\subsection{Step 3: Calculate Merging Scores}

Figure \ref{figure1} gives a graphical visualization of Step 3 in the 2-dimensional case. The number of points \(m_1, m_2, m_3\) are counted inside the three consecutive and adjacent hyper-cylinders (in this 2D example, rectangles) centered at \(\mu_{k_1}, \frac{\mu_{k_1} + \mu_{k_2}}{2}, \mu_{k_2}\), respectively. All of these 3 hyper-cylinders have central axis along \(l_{\mu_{k_1} \mu_{k_2}}\), base radius equal to \(r\), and height equal to \(||\frac{\mu_{k_2} - \mu_{k_1}}{2}||\). 

\begin{figure}[h]
    \centering
    \includegraphics[width=0.7\linewidth]{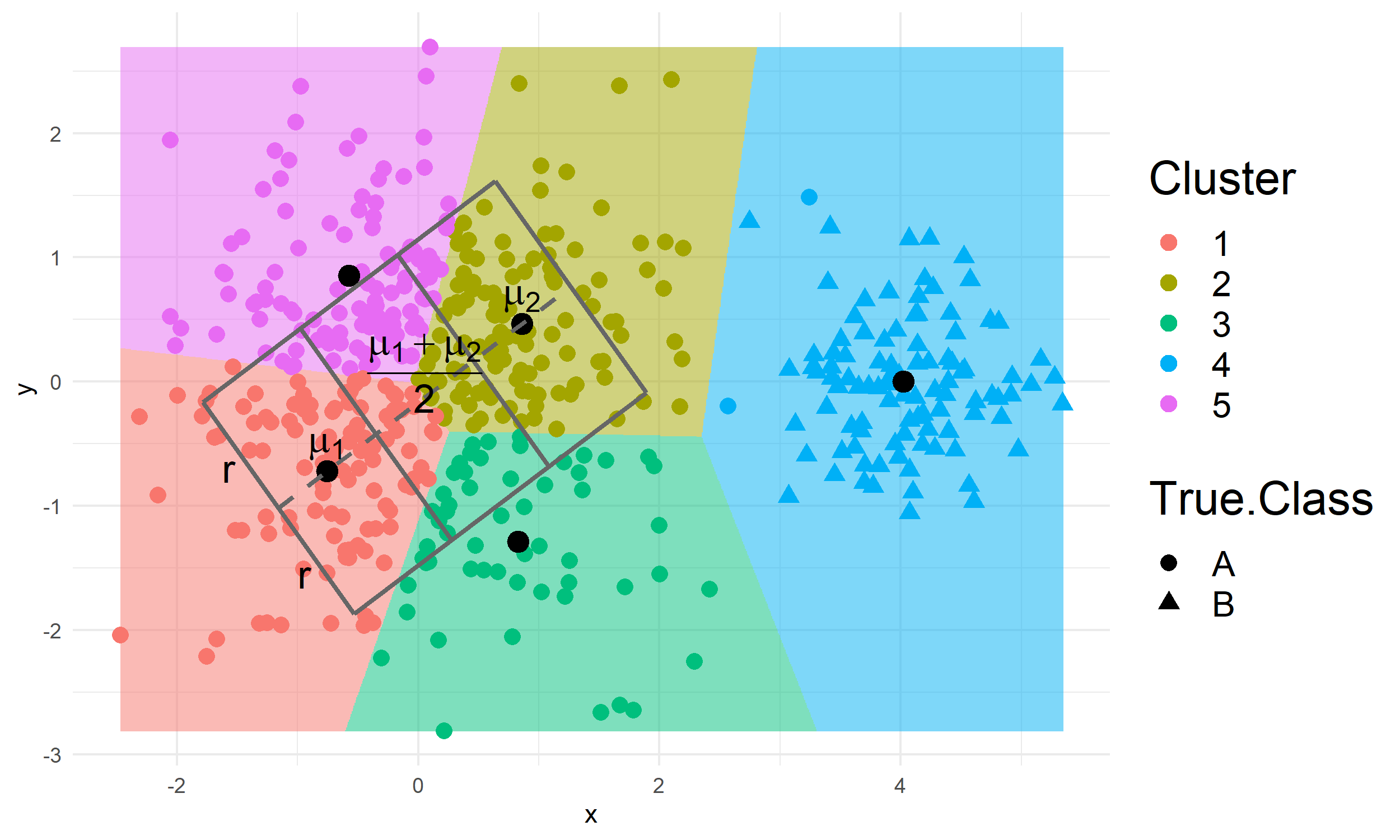}
    \caption{2D illustration of step 3.}
    \label{figure1}
\end{figure}

These three hyper-cylinders are of the same volume, so the ratio of points can be viewed as a representation of density ratios. This ratio further reflects the ``log concavity score'' of each pair of clusters, where a larger score indicates more evidence of merging this cluster pair.

\subsection{Step 4: Corner Cases}

Since our algorithm is based on the ratio of number of points, an obvious potential issue is the corner cases of initial clusters with very few observations. A cluster with very few observations indicates extremely sparse density in this subspace and is likely to return very low merging scores with all other clusters. But such a sparse cluster usually does not serve as a separate true class.

To accommodate this issue, we identify these extremely sparse clusters (\(|C_k| \leq 3\)), and link each one of them to its closest neighboring cluster.

\subsection{Step 5: Final Cluster Assignments}

As shown in Corollary 1, if the local log concavity assumption holds, theoretically we should set the threshold value \(t = 1\) and merge all cluster pairs with score greater than \(1\). However, this is usually not ideal in real-life applications with limited sample size or if the log-concavity assumption is slightly violated. 

So instead, when the cluster number is known, we use a hierarchical-clustering-type method. Note that this final merging step can be viewed as a hierarchical clustering using inverse scores \(\frac{1}{s_{k_1k_2}}\) as the distance metric and single linkage as the link function. Thus when the desired number of clusters is unknown, we construct a hierarchical clustering dendrogram using the inverse scores and cut the tree at the desired number of clusters.

\section{Performance Evaluation}
\label{sec:experiments}

In this section, we evaluate and contrast the clustering results produced by our approach with those of several competing methods. We first tested these methods on some artificial 2-dimensional datasets, previously analyzed by the competing methods in their respective studies. Then we employed these methods to some real-world, high-dimensional benchmark datasets to examine their efficiencies.

We compare the following 6 methods:

\begin{enumerate}
    \item \emph{CavMerge}: K-means merging based on log-concavity (our proposed method)

    \item \emph{SK}: Skeleton clustering \citep{Wei2023}

    \item \emph{DBSCAN}: Density-based spatial clustering of applications with noise \citep{Ester1996}

    \item \emph{k-mH}: K-means hierarchical clustering \citep{Peterson2018}
    
    \item \emph{KNOB}: KNOB-Syncytial Clustering \citep{AlmodovarRivera2020}

    \item \emph{HC}: Hierarchical clustering with Euclidean distance and single linkage
\end{enumerate}

Note that all 6 methods are capable of forming non-convex, general-shaped clusters. 

The clustering performance is evaluated by the \emph{Adjusted Rand Index (ARI)} \citep{Hubert1985} between the clustering results and the true labels. \emph{ARI} takes a value within \([0, 1]\), with the larger value indicating a better match with the true label. An \emph{ARI} of 1 indicates a perfect match. (Since \emph{ARI} makes some adjustments to the raw Rand Index, it may take values of small negative numbers, e.g., -0.002.)

\subsection{Two-Dimensional Datasets}

Here we evaluate the clustering performances of the 6 methods on 15 two-dimensional simulated datasets. Figure \ref{figure2} is a visualization of these datasets, where observations from different true classes are marked in different colors. These datasets have very different characteristics, for example convex/non-convex, varying densities, etc. Their sample sizes range from \(\sim 200\) to \(\sim 5000\), and the number of true classes range from \(2\) to \(8\). All these datasets have been used by one or more competing methods to demonstrate the effectiveness of their methods. Datasets \(1 \sim 7\) and \(9 \sim 14\) were studied by \citet{AlmodovarRivera2020}; \(2, 3, 4, 5, 10\) were studied by \citet{Peterson2018}, and \(8\) and \(15\) were studied by \citet{Wei2023}.

\begin{figure}[h]
    \centering
    \includegraphics[width=0.9\linewidth]{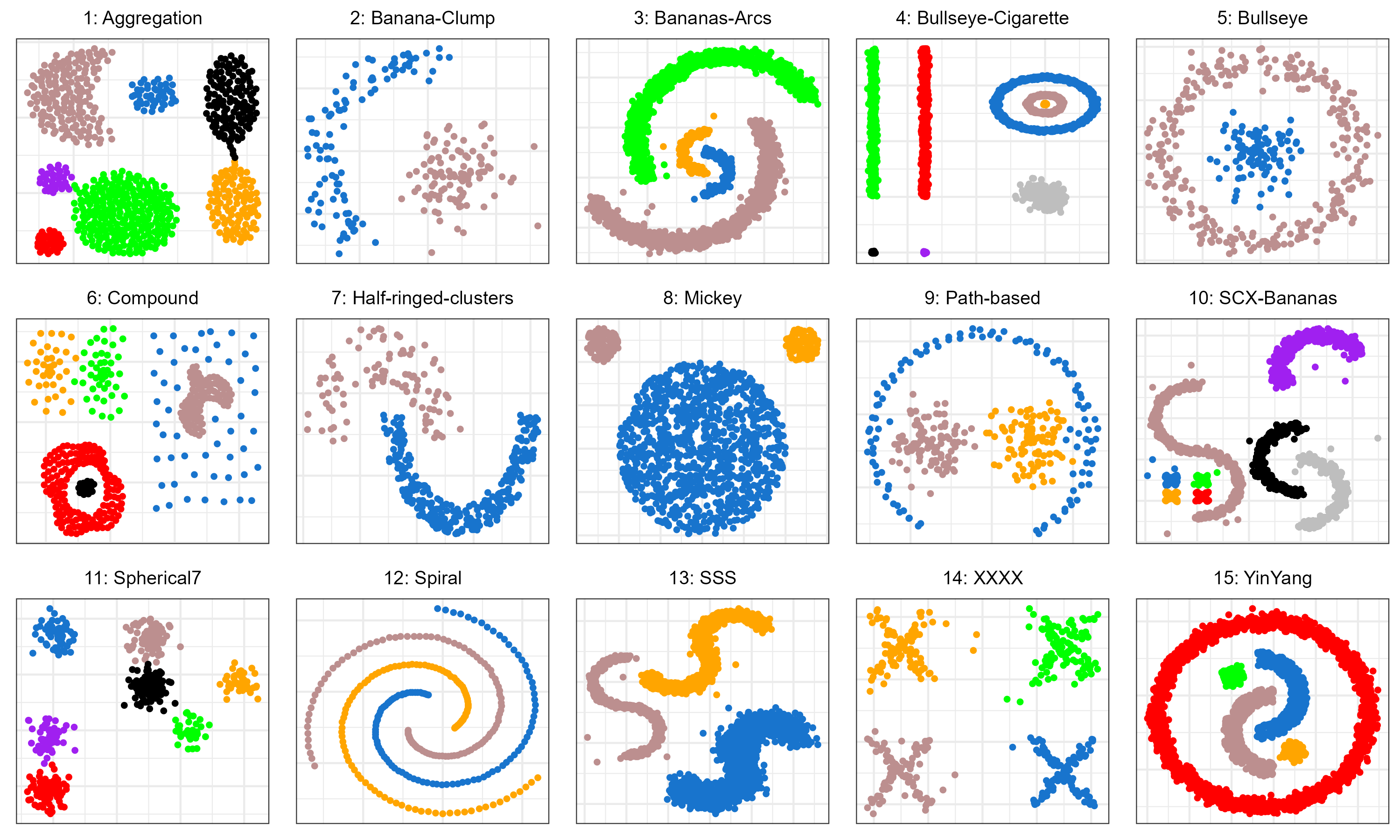}
    \caption{Fifteen 2D datasets used for performance evaluations.}
    \label{figure2}
\end{figure}

We ran each method on each dataset over 100 random seeds and measured the average performance. The number of clusters for \emph{k-mH} and \emph{KNOB} are automatically determined and cannot be set manually. For the rest 4 methods, we set the number of clusters equal to the true number of classes. To make a fair comparison among the merging methods based on K-means results (\emph{CavMerge, SK, k-mH, KNOB}), we used the same K-means clustering result (determined by the jump statistic introduced in Section \ref{sec:algo}, with 25 random starts and power \(\frac{2}{2} = 1\)) as the initial clusters for all these 4 methods. The \emph{k-mH} and \emph{DBSCAN} methods will identify some observations as ``outliers'', and we treat those as clusters of singletons according to \citep{Maitra2009}. All parameters are set to the default values in the corresponding R packages.

\begin{table}[h]
  \caption{Average \emph{ARI} of 100 random trials for each method.}
  \label{tab:1}
  \vspace{1em}
  \centering
  \newcolumntype{+}{>{\global\let\currentrowstyle\relax}}
  \newcolumntype{^}{>{\currentrowstyle}}
  \newcommand{\rowstyle}[1]{\gdef\currentrowstyle{#1}%
  #1\ignorespaces
}
  \newcolumntype{C}[1]{>{\centering\let\newline\\\arraybackslash\hspace{0pt}}m{#1}}
\begin{tabular}{+C{3cm}^C{1.8cm}^C{1.2cm}^C{1.5cm}^C{1.2cm}^C{1.2cm}^C{1.2cm}}
\hline
\rowstyle{\bfseries}
Data & CavMerge & SK & DBSCAN & k-mH & KNOB & HC \\
\hline
Aggregation & 0.990\newline (0.013) & 0.841\newline (0.098) & 0.809\newline (0) & 0.884\newline (0.102) & 0.812\newline (0.042) & 0.804\newline (0) \\
Banana-Clump & 0.565\newline (0.481) & 1\newline (0) & 0.980\newline (0) & 0.568\newline (0.222) & 0.674\newline (0.402) & 0.000\newline (0) \\
Bananas-Arcs & 0.973\newline (0.084) & 0.983\newline (0.044) & 0.981\newline (0) & 0.707\newline (0.331) & 0.837\newline (0.133) & 0.000\newline (0) \\
Bullseye & 0.992\newline (0.108) & 0.992\newline (0.005) & 0\newline (0) & 0.207\newline (0.197) & 0.652\newline (0.281) & -0.003\newline (0) \\
Bullseye-Cigarette & 0.978\newline (0.021) & 0.938\newline (0.021) & 0.958\newline (0) & 0.282\newline (0.072) & 0.874\newline (0.072) & 1\newline (0) \\
Compound & 0.754\newline (0.109) & 0.809\newline (0.011) & 0.764\newline (0) & 0.467\newline (0.149) & 0.710\newline (0.087) & 0.742\newline (0) \\
Half-ringed-clusters & 0.485\newline (0.362) & 0.270\newline (0.077) & 0.937\newline (0) & 0.0786\newline (0.052) & 0.863\newline (0.181) & 0.256\newline (0) \\
Mickey & 1\newline (0) & 0.634\newline (0.466) & 1\newline (0) & 0.347\newline (0.437) & 0.534\newline (0.292) & 1\newline (0) \\
Path-based & 0.425\newline (0.053) & 0.500\newline (0.135) & 0\newline (0) & 0.218\newline (0.166) & 0.470\newline (0.078) & 0.001\newline (0) \\
SCX-Bananas & 0.974\newline (0.087) & 0.952\newline (0.022) & 0.661\newline (0) & 0.774\newline (0.264) & 0.787\newline (0.149) & 0.171\newline (0) \\
Spherical7 & 1\newline (0) & 1\newline (0) & 0.630\newline (0) & 0.850\newline (0.156) & 1\newline (0) & 0.632\newline (0) \\
Spiral & 0.033\newline (0.016) & 0.040\newline (0.020) & 0\newline (0) & 0.018\newline (0.006) & 0.091\newline (0.072) & 1\newline (0) \\
SSS & 0.958\newline (0.136) & 1\newline (0) & 0.965\newline (0) & 0.833\newline (0.180) & 0.999\newline (0.004) & 0.000\newline (0) \\
XXXX & 0.913\newline (0.164) & 0.999\newline (0) & 0.984\newline (0) & 0.613\newline (0.321) & 0.993\newline (0.009) & 1\newline (0) \\
YinYang & 1\newline (0) & 0.775\newline (0.224) & 0.995\newline (0) & 0.810\newline (0.302) & 0.623\newline (0.293) & 1\newline (0) \\
\hline
\textbf{Overall} & \textbf{0.803} & \textbf{0.733} & \textbf{0.711} & \textbf{0.510} & \textbf{0.750} & \textbf{0.507}\\
\hline
\end{tabular}
\end{table}

Table \ref{tab:1} shows the performance of these six methods. We report the average \emph{ARI} over all 100 trials for each method in each dataset. We also provide visualization for clustering results from one random trial of our method \emph{CavMerge} in Appendix~\ref{app:vis_2d}. 

As we can see from Table \ref{tab:1}, our method has the highest overall average \(ARI = 0.803\), compared to the second best \emph{KNOB} with \(ARI = 0.750\). Apart from a few datasets, our method consistently produces similar or higher performance than all competing methods.

We also provide the average computational time for each method in Appendix~\ref{app:comp_time}. The computational time for our method is on the same scale with \emph{SK} and \emph{HC}, and is more than 10 times faster than \emph{KNOB} and \emph{k-mH}.

\subsection{High-Dimensional Datasets}

We are also interested in the performance of our method in real-world high-dimensional datasets. In this section, we test the 6 clustering algorithms introduced in the previous section on 5 benchmark datasets. The Olive Oils data can be accessed through R package \href{https://cran.r-project.org/web/packages/FlexDir/index.html}{FlexDir} and the rest 4 datasets can be found in the \href{https://archive.ics.uci.edu/ml/index.php}{UCI machine learning repository}. All these 5 datasets have been used by existing papers to evaluate their clustering performance. A detailed introduction of each dataset can be found in Appendix~\ref{app:data_desc}.

The hyper-parameters (if any) of each clustering method are set to be the default values given in their corresponding R packages; the number of minimum neighboring points for \emph{DBSCAN} is set to be 5 as it seems to be a reasonable choice regarding the sample size, and results have relatively high clustering performance. For methods that utilize merging of K-means results, we set the initial K-means number of starts equal to 100. The number of clusters for \emph{k-mH} and \emph{KNOB} are automatically determined and cannot be set manually. For the rest 4 methods, we set the number of clusters equal to the true number of classes.

Ten random trials of each clustering method were performed on each dataset, and the mean and standard error of the 10 \emph{ARI} values were reported. The clustering performance is presented in Table \ref{tab:2}. The best performer of each dataset is marked in blue, and the runner-up is marked in red.

\begin{table}[h]
  \caption{Average \emph{ARI} and standard error on benchmark data.}
  \label{tab:2}
  \bigbreak
  \centering
  \newcolumntype{+}{>{\global\let\currentrowstyle\relax}}
  \newcolumntype{^}{>{\currentrowstyle}}
  \newcommand{\rowstyle}[1]{\gdef\currentrowstyle{#1}%
  #1\ignorespaces
}
  \newcolumntype{C}[1]{>{\centering\let\newline\\\arraybackslash\hspace{0pt}}m{#1}}
  \begin{tabular}{+C{1.5cm}^C{1.8cm}^C{1.5cm}^C{1.5cm}^C{1.5cm}^C{1.5cm}^C{1.5cm}}
    \hline
    \rowstyle{\bfseries}
    Data & CavMerge & SK & DBSCAN & k-mH & KNOB & HC \\
    \hline
    Digits & \textcolor{blue}{0.720}\newline (0.046) & \textcolor{red}{0.551}\newline (0.082) & 0.048\newline (0) & -\newline (-) & 0.128\newline(0.098) & 0\newline (0) \\
    \hline
    Olive Oils & \textcolor{blue}{0.637}\newline (0.078) & \textcolor{red}{0.557}\newline (0.094) & -0.002\newline (0) & 0.210\newline (0.012) & 0.125\newline(0.161) & -0.001\newline (0) \\
    \hline
    Ecoli & \textcolor{red}{0.685}\newline (0.086) & \textcolor{blue}{0.722}\newline (0.068) & 0.047\newline (0) & 0.420\newline (0.194) & 0.249\newline(0.031) & 0.041\newline (0) \\
    \hline
    Iris & \textcolor{blue}{0.589}\newline (0.097) & 0.579\newline (0.054) & 0.564\newline (0) & 0.507\newline (0.011) & \textcolor{red}{0.587}\newline(0.060) & 0.564\newline (0) \\
    \hline
    Seeds & \textcolor{red}{0.377}\newline (0.171) & 0.222\newline (0.022) & 0\newline (0) & 0.142\newline (0.102) & \textcolor{blue}{0.512}\newline(0.159) & 0.002\newline (0) \\
    \hline
    \textbf{Overall} & \textbf{\textcolor{blue}{0.602}} & \textbf{\textcolor{red}{0.526}} & \textbf{0.131} & \textbf{0.320} & \textbf{0.320} & \textbf{0.121}\\
    \hline
  \end{tabular}
  \caption*{(All 10 trials of \emph{km-H} failed for the Digits dataset)}
\end{table}

Our proposed \emph{CavMerge} algorithm achieves the best performance in 3 of the 5 datasets, and is the second best in the rest 2 datasets. Specifically, \emph{CavMerge} achieves a significantly higher clustering accuracy than all other methods in the handwritten digits dataset, which is the dataset with the largest sample size (\(>5000\)) and feature dimension (64), and is considered the most difficult one among the 5 datasets.

\section{Conclusions and Discussions}
\label{sec:conc}

This paper introduces a novel clustering method for merging K-means clusters based on log-concavity. This approach is characterized by its minimal distributional assumptions, intuitive geometric rationale, and efficient computational process. It exhibits enhanced performance in both simulated 2-dimensional datasets as well as real-world, higher-dimensional datasets.

We provide some detailed discussion and possible extensions below.

\subsection{Comparison to Skeleton Clustering}

\citet{Wei2023} introduced skeleton clustering, a syncytial clustering method using a framework similar to our method. As shown in Section~\ref{sec:experiments}, our method achieves better clustering performance than skeleton clustering in most datasets. Here we briefly state the difference between the two methods.

The skeleton clustering algorithm estimates the density (vertice/face/tube) between each pair of cluster centers by some kernel function, obtains a pairwise density matrix, and finds a cutoff to merge clusters with large pairwise density. Since it involves kernel smoothing, many hyper-parameters need to be determined or tuned before performing the analysis. Also, if the densities differ significantly in different true classes, then merging in a high-density class will be much more favored than in a low-density class.

On the other hand, our algorithm uses the log-concave property and computes a score at each decision boundary using density ratios. Compared to the pairwise density measure in \citet{Wei2023}, this score is completely non-parametric and scale-free. All pairwise scores are on the same scale, and a global comparison among the scores is fair. A more detailed discussion and some examples are presented in Appendix~\ref{app:compare_sk}.

\subsection{Discussions and Extensions}

Our proposed method is based on point counting, thus the sample size has a big impact on the clustering performance. When the sample size is small, for instance, \(n < p^2\), we suggest to combine our method with bootstrapping methods on the initial K-means clusters, e.g. \citet{Maitra2012}.

Note that although we use K-Means clustering as the initial cluster assignment in this paper, our method is not restricted to K-means but applies to other clustering methods, especially those that form convex clusters. For instance, our method can be applied directly to the K-Means clustering with feature selection such as \citet{Witten2010} and \citet{Zhang2020}, simply by only keeping the selected features when calculating the merging score. Another straightforward extension is to combine this log-concave property with model-based clustering algorithms. For instance, \citet{Baudry2010} introduced a clustering method based on combining Gaussian model components, which fits perfectly into our framework.

Apart from cluster analysis, this idea of merging based on log-concavity can also be used in other fields, for instance determining the number of mixing components in the context of finite mixture modeling \citep{Manole2021}.

\section*{Acknowledgment}

[To be added.]


\bigskip
\begin{center}
{\large\bf SUPPLEMENTAL MATERIALS}
\end{center}

\begin{description}

\item[R code and data:] R package implementing the \emph{CavMerge} algorithm, along with all simulation datasets used in Section~\ref{sec:experiments}. (GNU zipped tar file)

\end{description}


\appendix

\section{Proof of Theorem 1}
\label{app:proof_thm1}

We restate Theorem 1 here and provide its proof: 

\paragraph{Theorem 1.} Assume that a \(p\)-dimensional random variable \(X\) has log concave pdf \(f(\cdot)\), and let \(x_1, ..., x_n \in R^p\) denote a realization of this distribution with \(n\) observations. For any subspace \(Q \subset R^p\) and any vector \(a \in R^p\), let

\begin{align*}
    m_1 = \sum^n_{i=1} 1(x_i+a \in Q),\quad m_2 = \sum^n_{i=1} 1(x_i \in Q),\quad m_3 = \sum^n_{i=1} 1(x_i-a \in Q)
\end{align*}

denote the number of observations in subspaces \(Q-a, Q, Q+a\), respectively. Further, assume that the cumulative densities on those three subspaces are all positive. 

\begin{enumerate}
    \item[(a)] For any given \(\epsilon,\delta > 0\) and if \(m_1 m_3 > 0\), there exists a \(N\) such that for all \(n>N\), we have \(P(\frac{m^2_2}{m_1 m_3} > 1-\epsilon) > 1-\delta\)
    
    \item[(b)] The convergence rate of \(\frac{m^2_2}{m_1 m_3}\) to \(C = \frac{(\int_{U} f(x) dx)^2}{\int_{U-a} f(x) dx \int_{U+a} f(x) dx}\) is \(O(\frac{1}{\sqrt{n}})\).
\end{enumerate}

\begin{proof}

Let random variables \(U_{1i} = 1(X_i+a \in U), U_{2i} = 1(X_i \in U), U_{3i} = 1(X_i-a \in U)\) be the indicator of \(X\) falling in \(U-a, U, U+a\), respectively. By the Weak Law of Large Numbers:

\begin{align*}
    \bar U_{1n} &\overset{p}{\to} \int_{U-a} f(x) dx \equiv P_1\\
    \bar U_{2n} &\overset{p}{\to} \int_{U} f(x) dx \equiv P_2\\
    \bar U_{3n} &\overset{p}{\to} \int_{U+a} f(x) dx \equiv P_3\\
    \frac{\bar U^2_{2n}}{\bar U_{1n} \bar U_{3n}} &= \frac{m^2_2 / n^2}{m_1/n \cdot m_3 / n} = \frac{m^2_2}{m_1 m_3} \\
    &\overset{p}{\to} \frac{P^2_2}{P_1 P_3} \equiv C
\end{align*}

where \(\bar U_{kn} = \frac{\sum^n_{i=1} U_{ki}}{n}, k = 1,2,3\). The last convergence follows from the product rule for convergence in probability.

Since \(f(\cdot)\) is log concave, we have \(\log f(x-a) + \log f(x + a) \leq 2 \log f(x), \forall a \in R^p\). Taking the exponential and simply applying a double integral over \(x \in U\) on both sides yields \(\int_{U-a} f(x) dx \int_{U+a} f(x) dx \leq (\int_{U} f(x) dx)^2\), i.e., \(C = \frac{P^2_2}{P_1 P_3} \geq 1\). 

Thus for any \(\epsilon, \delta > 0\), there exists a \(N\) such that for all \(n > N\),

\begin{align*}
    P(|\frac{\bar U^2_{2n}}{\bar U_{1n} \bar U_{3n}} - C| > \epsilon) &= P(|\frac{m^2_2}{m_1 m_3} - C| > \epsilon) < \delta\\
    P(\frac{m^2_2}{m_1 m_3} > 1-\epsilon) &\geq P(\frac{m^2_2}{m_1 m_3} > C-\epsilon) \\
    &= P(\frac{m^2_2}{m_1 m_3} - C > -\epsilon)\\
    &\geq P(|\frac{m^2_2}{m_1 m_3} - C| < \epsilon) \\
    &> \delta
\end{align*}

This is the end of the proof for part (a).

For part (b), similar to the proof of part (a), we can get 

\[\frac{m_1}{n} \frac{m_3}{n} = \bar U_{1n} \bar U_{3n} \overset{p}{\to} \int_{U-a} f(x) dx \int_{U+a} f(x) dx \equiv P_1 P_3
\]

and thus \(\frac{m^2_2}{m_1 m_3} = \frac{\bar U_{2n}}{\bar U_{1n} \bar U_{3n}} \overset{p}{\to} \frac{\bar U^2_{2n}}{P_1 P_3}\). As \(n \to \infty\) and when \(m_1 m_3 \neq 0\):

\begin{align*}
    P(|\frac{m^2_2}{m_1 m_3} - C| > \epsilon) &\approx  P(|\frac{\bar U^2_{2n}}{P_1 P_3} - C| > \epsilon) = P(|\frac{\bar U^2_{2n}}{P_1 P_3} - \frac{P^2_2}{P_1 P_3}| > \epsilon) = P(|\frac{\bar U^2_{2n} - (E \bar U_{2n})^2}{P_1 P_3}| > \epsilon)\\
    &\leq P(|\bar U^2_{2n} - (E \bar U_{2n})^2| > \epsilon)\\
    &\leq P(|\bar U_{2n} - E \bar U_{2n}| > \frac{\epsilon}{2}) \\
    &= P(|\sum^n_{i=1} 1(X_i \in U) - E[\sum^n_{i=1} 1(X_i \in U)]| > \frac{n \epsilon}{2})\\
    &\leq 2 \exp (-2\frac{(n \epsilon / 2)^2}{n}) = 2 \exp (-\frac{n \epsilon^2}{2})
\end{align*}

The first approximation is the direct result of \(\frac{m^2_2}{m_1 m_3} \overset{p}{\to} \frac{\bar U^2_{2n}}{P_1 P_3}\);

The first \(\leq\) follows from the fact that \(P_1 P_3 \in (0, 1]\);

The second \(\leq\) comes from the fact that \(\bar U_{2n}\) is bounded within \([0, 1]\), thus \(|\bar U_{2n} + E \bar U_{2n}| \in (0, 2]\);

The third \(\leq\) follows directly from Hoeffding's inequality.

Thus the convergence rate of \(\frac{m^2_2}{m_1 m_3}\) to \(C = \frac{P^2_2}{P_1 P_3}\) is \(O(\frac{1}{\sqrt{n}})\):

\begin{align*}
    \delta &\leq 2 \exp (-\frac{n \epsilon^2}{2})\\
    \epsilon &\leq \sqrt{\frac{2}{n} (\log 2 - \log \delta)}\\
    |\frac{m^2_2}{m_1 m_3} - C| &= O_p(\frac{1}{\sqrt{n}})
\end{align*}

\end{proof}

\section{Proof of Proposition~\ref{prop:linear_adjacent} and Detailed Explanation of Step~2}
\label{app:proof_prop1}

This section provides a complete proof of the linear adjacency bound (Proposition~\ref{prop:linear_adjacent}) and gives a full justification for the approximate identification strategy used in Step~2 of \emph{CavMerge}.

\paragraph{Background: Voronoi diagrams and Delaunay triangulations.}

The Voronoi diagram of the \(K\) cluster centers \(\mu_1, \ldots, \mu_K\) partitions \(\mathbb{R}^p\) into \(K\) convex Voronoi cells, where the cell of \(\mu_k\) is
\[
V_k = \{x \in \mathbb{R}^p : \|x - \mu_k\| \leq \|x - \mu_{k'}\|, \; \forall k' \neq k\}.
\]
Two cells \(V_{k_1}\) and \(V_{k_2}\) share a \((p-1)\)-dimensional boundary---i.e., the pair \((C_{k_1}, C_{k_2})\) is adjacent---if and only if the centers \(\mu_{k_1}\) and \(\mu_{k_2}\) are connected by an edge in the \emph{Delaunay triangulation} \(\mathrm{Del}(\{\mu_1, \ldots, \mu_K\})\), the dual graph of the Voronoi diagram. Equivalently, \(\mu_{k_1}\) and \(\mu_{k_2}\) are Delaunay neighbors if and only if there exists an empty sphere (a sphere whose open interior contains no cluster center) passing through both \(\mu_{k_1}\) and \(\mu_{k_2}\). In high-dimensional spaces, computing this triangulation exactly is intractable \citep{Polianskii2020}; Step~2 therefore uses an efficient approximation described at the end of this section.

\paragraph{The naive bound and its inadequacy.}

The trivial upper bound on the number of adjacent pairs is \(\frac{K(K-1)}{2}\), the total number of distinct cluster pairs. In full generality, the Delaunay triangulation of \(K\) points in \(\mathbb{R}^p\) can have up to \(O(K^{\lceil p/2 \rceil})\) simplices (McMullen's Upper Bound Theorem). This worst case is achieved only when the points lie on a one-dimensional moment curve---a highly degenerate configuration that does not arise in practice. The cluster centers in \emph{CavMerge} are empirical means of data subsets, and as such inherit the geometric structure of the data. The following proposition formalizes the resulting improvement.

\paragraph{Proposition~\ref{prop:linear_adjacent} (Linear adjacency bound, restated).}

\begin{proposition}
Under the manifold hypothesis that the data concentrate near a \((p-1)\)-dimensional submanifold of \(\mathbb{R}^p\), the number of adjacent cluster pairs equals the number of edges in the Delaunay triangulation of the \(K\) cluster centers, which is \(O(K)\).
\end{proposition}

\begin{proof}
The \(K\) cluster centers, each being the mean of a subset of data points drawn from a distribution near a submanifold, lie approximately on that submanifold. Under the manifold hypothesis, this submanifold is a \((p-1)\)-dimensional polyhedral surface \(S \subset \mathbb{R}^p\), and the \(K\) centers form an approximate sparse sample of \(S\). We establish the \(O(K)\) bound in two cases.

\medskip
\textbf{Case 1: \(p = 2\).}

The Delaunay triangulation of any finite set of \(K\) points in the plane \(\mathbb{R}^2\) is a planar graph. By Euler's polyhedral formula, a simple planar graph on \(K\) vertices and \(E\) edges satisfies \(K - E + F = 2\), where \(F\) is the number of faces. Since each face has at least 3 bounding edges and each edge borders at most 2 faces, we get \(2E \geq 3(F-1)\), which combined with Euler's formula gives \(E \leq 3K - 6\). Therefore, the number of adjacent cluster pairs is at most \(3K - 6 = O(K)\), regardless of any distributional assumption.

\medskip
\textbf{Case 2: \(p \geq 3\).}

We invoke two classical results from computational geometry.

\medskip
\textit{Result 1 (Attali \& Boissonnat, 2004).} Let \(S \subset \mathbb{R}^p\) be a fixed polyhedral surface, i.e., a finite union of interior-disjoint planar facets of total area \(A_S\), total boundary perimeter \(L_S\), and \(C_S\) facets. Consider \(K\) points forming an \((\varepsilon, \kappa)\)-sample of \(S\), meaning: for every point \(x\) on any facet \(F\) of \(S\), the ball \(B(x, \varepsilon)\) contains at least one sample point from \(F\), and no ball \(B(x, 2\varepsilon)\) contains more than \(\kappa\) sample points from \(F\). Under this condition, \citet{AttaliLinear2004} prove that the total number of Delaunay edges of the \(K\) sample points is at most
\begin{align}
    \left(1 + \frac{C_S \kappa}{2} + 5300\pi\kappa^2 \frac{L_S^2}{A_S}\right) K. \label{eq:attali_bound}
\end{align}

The constant in \eqref{eq:attali_bound} depends only on the fixed surface \(S\) (through \(C_S\), \(A_S\), \(L_S\)) and the sampling density ratio \(\kappa\), but \emph{not} on \(K\). Hence the bound is of the form \(c \cdot K\) for a constant \(c > 0\), i.e., \(O(K)\).

To understand why this bound is linear, \citet{AttaliLinear2004} decompose the sample points into two zones. In the \(\varepsilon\)-\emph{regular zone} (points in the interior of the facets, at distance \(> \varepsilon\) from any edge), each sample point \(x\) has at most \(C_S \kappa\) Delaunay neighbors: any empty sphere through \(x\) intersects the supporting plane of each facet in a disk of radius at most \(\varepsilon\), which by the sampling condition contains at most \(\kappa\) sample points. The \(\varepsilon\)-\emph{singular zone} (points near edges of \(S\)) contains at most \(O(\kappa L_S / (\varepsilon \sqrt{A_S}) \cdot \sqrt{K}) = O(\sqrt{K})\) sample points, and their Delaunay edges contribute at most \(O(K)\) edges in total. Summing both contributions yields the linear bound \eqref{eq:attali_bound}.

\medskip
\textit{Result 2 (Amenta, Attali \& Devillers, 2007).} For the more general case of a \(q\)-dimensional polyhedron \(\mathbb{P} \subset \mathbb{R}^p\) with \(2 \leq q \leq p-1\), \citet{Amenta2007} prove that the Delaunay triangulation of \(K\) points forming a sparse \(\varepsilon\)-sample of \(\mathbb{P}\) has complexity
\begin{align}
    O\!\left(K^{(p-1)/q}\right). \label{eq:amenta_bound}
\end{align}
The key structural insight is that the \emph{medial axis} of any \(q\)-dimensional polyhedron in \(\mathbb{R}^p\) (the set of points equidistant from two or more points on \(\mathbb{P}\)) is always \((p-1)\)-dimensional, regardless of the dimension \(q\) of \(\mathbb{P}\). Since the number of Delaunay balls is governed by the medial axis dimension (\(p-1\)), while the number of sample points \(K\) grows as \(\varepsilon^{-q}\) (determined by the polyhedron dimension \(q\)), one obtains the relation \(|\text{Delaunay balls}| = O(\varepsilon^{-(p-1)}) = O(K^{(p-1)/q})\).

When \(q = p - 1\)---the case where the data manifold is a \((p-1)\)-dimensional hypersurface---the exponent becomes \((p-1)/(p-1) = 1\), and the bound \eqref{eq:amenta_bound} reduces to \(O(K)\). This is precisely the setting of Proposition~\ref{prop:linear_adjacent}: data concentrated near a \((p-1)\)-dimensional submanifold of \(\mathbb{R}^p\).

\medskip
Combining both cases, the number of adjacent cluster pairs is \(O(K)\) whenever the cluster centers concentrate near a \((p-1)\)-dimensional submanifold of \(\mathbb{R}^p\).
\end{proof}

\paragraph{Remark on lower-dimensional manifolds.}

If the data lie near a manifold of intrinsic dimension \(q < p-1\), the bound \eqref{eq:amenta_bound} gives \(O(K^{(p-1)/q})\). This is still strictly better than the quadratic bound \(O(K^2)\) whenever \(q > (p-1)/2\). In real datasets, even when the feature dimension \(p\) is large, the intrinsic dimension \(q\) of the data is typically much smaller than \(p\) but usually \(q \geq p/2\), making the bound sub-quadratic.

\paragraph{The approximate identification strategy and its justification.}

Since computing the exact Delaunay triangulation is intractable in high-dimensional spaces \citep{Polianskii2020}, Step~2 identifies adjacent pairs via the following approximation: for each observation \(x_i\), find its two nearest cluster centers \(\mu_{k_1(i)}\) and \(\mu_{k_2(i)}\), and declare \((C_{k_1(i)}, C_{k_2(i)})\) adjacent. The collection of all such pairs is
\[
    A = \left\{\bigcup_{i=1}^n (C_{k_1(i)}, C_{k_2(i)})\right\}.
\]
This approximation has two important properties:

\begin{itemize}
    \item \textbf{No false positives.} Suppose \((C_{k_1}, C_{k_2}) \in A\), so there exists some \(x_i\) with \(\mu_{k_1}\) and \(\mu_{k_2}\) as its two nearest centers. Then \(x_i\) lies in the region
    \[
        B_{k_1 k_2} = \bigl\{x \in \mathbb{R}^p : \max(\|x - \mu_{k_1}\|,\, \|x - \mu_{k_2}\|) < \|x - \mu_{k'}\|,\; \forall\, k' \neq k_1, k_2\bigr\}.
    \]
    The region \(B_{k_1 k_2}\) is nonempty if and only if \(C_{k_1}\) and \(C_{k_2}\) share a nontrivial decision boundary---i.e., they are adjacent. Therefore, no non-adjacent pair is ever included in \(A\).

    \item \textbf{Harmless false negatives.} If a truly adjacent pair \((C_{k_1}, C_{k_2})\) is not identified (no observation falls in \(B_{k_1 k_2}\)), then the shared Voronoi boundary between the two clusters contains no observed data points. In this situation, the two clusters almost certainly belong to different underlying distributions (there is a ``cavity'' between them in the data), and omitting their merging score does not harm the final result: a missing score for a between-distribution pair simply means that pair is not merged, which is the correct outcome.
\end{itemize}

Not every pair of clusters is adjacent, and calculating scores for all \(\frac{K(K-1)}{2}\) cluster pairs is redundant. The identification step above ensures that only \(O(K)\) adjacent pairs need to be evaluated in Step~3, making the overall algorithm computationally efficient even for large \(K\).

\section{Visualization of \emph{CavMerge} Results on 2D Data}
\label{app:vis_2d}

In Figure~\ref{supp_figure1} we provide the results from one random trial (seed = 1) of our method \emph{CavMerge} on the fifteen 2D datasets:

\begin{figure}[h]
    \centering
    \includegraphics[width=0.9\linewidth]{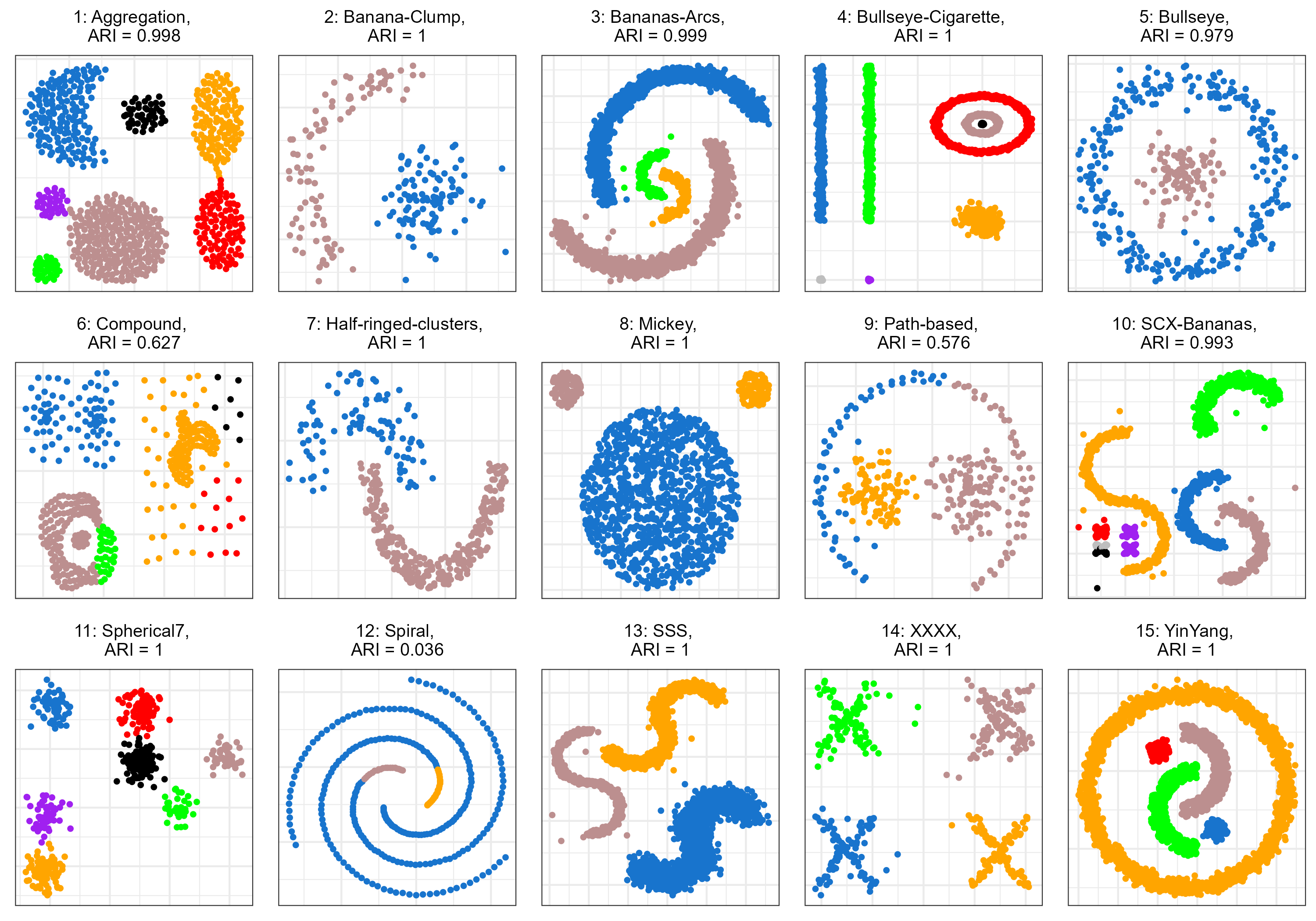}
    \caption{Performance of \emph{CavMerge} on fifteen 2D datasets.}
    \label{supp_figure1}
\end{figure}

\section{Computational Time for 2D Data}
\label{app:comp_time}

Table~\ref{tab:comp_time} shows the average computational time for each method on the fifteen 2D datasets. The computational time for \emph{CavMerge} algorithm is on the same scale with the \emph{Skeleton Clustering} and \emph{Hierarchical Clustering}, and is more than 10 times faster than \emph{KNOB} and \emph{k-mH}. 

(We also note that the \emph{SK} and \emph{k-mH} methods failed in some random trials of certain datasets. For both the computational time here and the average \emph{ARI} in Table~\ref{tab:1}, we only take into account the successful trials.)

\begin{table}[h]
  \caption{Average computational time of 100 random trials for each method.}
  \label{tab:comp_time}
  \vspace{1em}
  \centering
  \newcolumntype{+}{>{\global\let\currentrowstyle\relax}}
\newcolumntype{^}{>{\currentrowstyle}}
\newcommand{\rowstyle}[1]{\gdef\currentrowstyle{#1}%
  #1\ignorespaces
}
\begin{tabular}{+l^c^c^c^c^c^c}
    \hline
    \rowstyle{\bfseries}
    Data & CavMerge & SK & DBSCAN & k-mH & KNOB & HC \\
    \hline
    Aggregation & 0.209 & 0.195 & 0.015 & 30.4 & 3.39 & 0.017 \\
    Banana-Clump & 0.048 & 0.026 & 0.001 & 5.07 & 1.70 & 0.002 \\
    Bananas-Arcs & 0.427 & 0.349 & 0.019 & 31.2 & 14.4 & 0.627 \\
    Bullseye & 0.074 & 0.033 & 0.003 & 15.2 & 2.22 & 0.006 \\
    Bullseye-Cigarette & 0.368 & 0.214 & 0.015 & 16.8 & 11.1 & 0.350 \\
    Compound & 0.065 & 0.047 & 0.002 & 14.4 & 2.41 & 0.005 \\
    Half-ringed-clusters & 0.058 & 0.032 & 0.002 & 14.1 & 2.21 & 0.005 \\
    Mickey & 0.101 & 0.116 & 0.008 & 2.77 & 4.81 & 0.042 \\
    Path-based & 0.068 & 0.043 & 0.002 & 9.70 & 2.84 & 0.004 \\
    SCX-Bananas & 0.374 & 0.176 & 0.017 & 22.8 & 12.2 & 0.308 \\
    Spherical7 & 0.048 & 0.011 & 0.003 & 21.4 & 1.20 & 0.008 \\
    Spiral & 0.098 & 0.040 & 0.002 & 10.7 & 2.10 & 0.003 \\
    SSS & 1.70 & 0.407 & 0.026 & 9.59 & 15.8 & 0.865 \\
    XXXX & 0.070 & 0.034 & 0.003 & 16.1 & 2.06 & 0.005 \\
    YinYang & 0.335 & 0.203 & 0.017 & 16.6 & 10.3 & 0.330 \\
    \hline
    \textbf{Overall} & \textbf{0.270} & \textbf{0.118} & \textbf{0.009} & \textbf{15.785} & \textbf{6.234} & \textbf{0.172}\\
    \hline
\end{tabular}
\end{table}

\section{Data Description for Section~\ref{sec:experiments}}
\label{app:data_desc}

Below is a short introduction to each dataset used in Section~\ref{sec:experiments}:

\begin{enumerate}
    \item \textbf{Optical Recognition of Handwritten Digits.}

    The \href{https://archive-beta.ics.uci.edu/dataset/80/optical+recognition+of+handwritten+digits}{handwritten digits recognition data} \citep{Alpaydin1998} consists of 5620 pictures of handwritten digits by 43 different people. Each picture was first collected as a \(32 \times 32\) bitmap, then divided into non-overlapping blocks of \(4 \times 4\), and the number of pixels are counted in each block. In the end the input dimension of each observation is \(8 \times 8 = 64\) attributes, with each element taking values in \(\{1,2,...,16\}\). 

    \item \textbf{Olive Oils.}

    The olive oils dataset measures 8 different chemical components for 572 samples of olive oil. The samples are taken from 9 different areas in Italy: Coast-Sardinia, Inland-Sardinia, East-Liguria, West-Liguria, Umbria, Calabria, North-Apulia, South-Apulia, and Sicily. The olive oils dataset is a popular dataset for cluster analysis. It has been studied in \cite{Wei2023} and \cite{AlmodovarRivera2020}.

    \item \textbf{Ecoli.}

    The \href{https://archive-beta.ics.uci.edu/dataset/39/ecoli}{ecoli} dataset includes the identification of protein localization sites for the E.coli bacteria \citep{Jain1999, Nakai1992}. The original dataset contains 336 samples and 7 numerical attributes, with one attribute taking values in \{0.5, 1\} and over 99\% equal to 0.5. So we use the rest 6 attributes for clustering. The true labels have 8 categories. This dataset is also used in \cite{AlmodovarRivera2020}.

    \item \textbf{Iris.}

    The famous \href{https://archive-beta.ics.uci.edu/dataset/53/iris}{iris} data consists of 150 instances in 3 different species, each species having 50 samples. There are 4 numerical attributes: sepal length, sepal width, petal length, and petal width. 

    \item \textbf{Seeds.}

    The \href{https://archive-beta.ics.uci.edu/dataset/236/seeds}{seeds} data consists of the geometrical properties of kernels belonging to 3 different types of wheat: Kama, Rosa, and Canadian. The dataset has 7 numerical attributes and 210 instances, each type of wheat has 70 samples.
\end{enumerate}

\section{Comparison to Skeleton Clustering in \cite{Wei2023}}
\label{app:compare_sk}

After we developed our \emph{CavMerge} algorithm, we noticed that a recently proposed method \emph{Skeleton Clustering} by \citet{Wei2023} uses a similar structure to ours. Specifically, our definition for the three hyper-cylinders in Step~3 of our algorithm is related to the ``Tube Density'' as described in \cite{Wei2023}. However, instead of their direct estimation of densities via kernel smoothing, our method uses log concavity and measures the densities ratios via the ratio of points. In Section~\ref{sec:experiments}, we showed through some experimental studies that our algorithm achieves better performance. 

Here we give a short introduction about the \emph{Skeleton Clustering} method, and provide some comparison of differences between the two methods.

The method introduced in \citet{Wei2023} views the initial K-means clusters as the knots in a graph, and use several different methods to measure the edge weights between pairs of knots. Specifically, the ``Tube Density'' takes a tube-like hyper-cylinder areas between each pair of cluster centers, measures the minimum disk density (the marginal density on a 2D circle) along each tube, and constructs an adjacency matrix based on these density measures (The densities are estimated via local Gaussian kernel smoothing). Since it measures the minimum disk density, this technique is capable of identifying ``gaps'' (i.e., whether there is a space in the tube that only very few / no observations falls into) between two clusters that are faraway and should not be merged together, therefore provides reliable final clustering results.

There are two major differences between our \emph{CavMerge} algorithm and the \emph{Skeleton Clustering} algorithm:

\begin{itemize}
    \item The \emph{CavMerge} algorithm is purely non-parametric. 

    Our algorithm does ``nothing but counting points'', it involves no parameter estimate of any kind. The only unknown parameter in \emph{CavMerge} is the threshold for cutting the hierarchical tree, or equivalently, the number of clusters. If the number of final clusters is pre-determined, no hyper-parameter is involved in the \emph{CavMerge} algorithm.

    On the other hand, the \emph{Skeleton Clustering} involves tuning of several hyper-parameters, e.g., kernel bandwidth, rate adjustment for the bandwidth, disk radius, etc.

    \item The \emph{CavMerge} algorithm is scale-free.

    The \emph{Skeleton Clustering} estimates local density, thus these estimates depend on the level of absolute density of true distributions. They are also highly related to the choice of hyper-parameters. 
    
    On the other hand, our \emph{CavMerge} algorithm seeks to calculate a merging score based on density ratios. This score is unitless (ratio of densities) and can be compared globally. 
    
    In other words, if we consider a dataset with mixture of two distributions, one with high density and the other with low density, then the density estimates might vary in magnitude for cluster pairs in those two true classes.
\end{itemize}

Here we present one clustering result of the \emph{Aggregation} dataset: We performed K-means clustering with 28 clusters (determined by the jump statistic), and use the clustering results for the initial clusters of \emph{CavMerge} and \emph{Skeleton Clustering} algorithms. The results are shown in Figure~\ref{supp_figure2}: (a) shows the initial 28 clusters; (b) and (c) shows the final merging results of \emph{CavMerge} and \emph{Skeleton Clustering} algorithms.

\begin{figure}[h]
    \centering
    \includegraphics[width=0.9\linewidth]{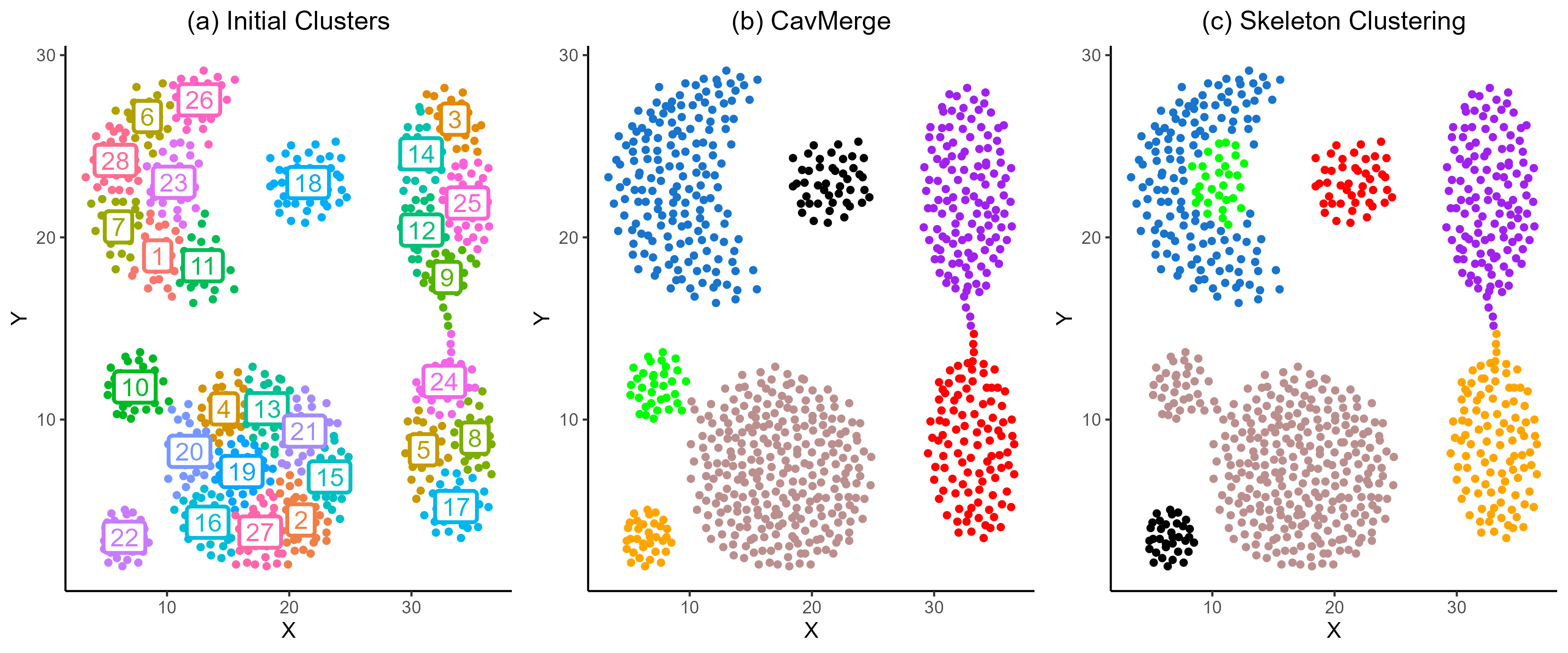}
    \caption{Visualization for: (a) 28 initial K-means clusters; (b) Merging results for \emph{CavMerge}; (c) Merging results for \emph{Skeleton Clustering}.}
    \label{supp_figure2}
\end{figure}

As we can see, the \emph{Skeleton Clustering} links cluster 10 with cluster 20, and separates cluster 23 from its neighbors. If we dig deeper into the scores (for \emph{CavMerge}) or weights (for \emph{Skeleton Clustering}) of these cluster pairs, we get:

\begin{itemize}
    \item $\min score\{(23, 1), (23, 6), (23, 7), (23, 11), (23, 26), (23, 28)\} = 0.721$

    \item $score(10, 20) = 0.22$

    \item $\max weight\{((23, 1), (23, 6), (23, 7), (23, 11), (23, 26), (23, 28))\} = 3.529$

    \item $weight(10, 20) = 5.439$
\end{itemize}

i.e., \emph{CavMerge} links cluster 23 to all its neighbors and separates cluster 10 from 20, while \emph{Skeleton Clustering} doesn't link cluster 23 to any of its neighbors but identifies cluster 10 and 20 as adjacent. We suspect this is probably due to the density difference between the true class of cluster 10 and the true class that cluster 20 belongs to.

\bibliographystyle{plainnat}
\bibliography{merging}

\end{document}